\documentclass[11pt]{article}

\usepackage[margin=1in]{geometry}
\usepackage{amsmath,amssymb,amsthm,mathtools}
\usepackage{aliascnt}
\usepackage{booktabs}
\usepackage{xcolor}
\usepackage{microtype}
\usepackage{enumitem}
\usepackage[numbers,sort&compress]{natbib}
\usepackage[colorlinks=true,linkcolor=blue!55!black,citecolor=blue!55!black,urlcolor=blue!55!black]{hyperref}
\hypersetup{
  pdftitle={Planning Against Learning in Rank-1 Games},
  pdfauthor={William Overman},
  pdfsubject={Strategic planning against learning dynamics in rank-one bimatrix games}
}
\usepackage[nameinlink,noabbrev]{cleveref}

\newtheorem{theorem}{Theorem}[section]

\newaliascnt{proposition}{theorem}
\newtheorem{proposition}[proposition]{Proposition}
\aliascntresetthe{proposition}

\newaliascnt{lemma}{theorem}
\newtheorem{lemma}[lemma]{Lemma}
\aliascntresetthe{lemma}

\newaliascnt{corollary}{theorem}
\newtheorem{corollary}[corollary]{Corollary}
\aliascntresetthe{corollary}

\newaliascnt{definition}{theorem}
\newtheorem{definition}[definition]{Definition}
\aliascntresetthe{definition}

\newaliascnt{remark}{theorem}

\aliascntresetthe{remark}

\newaliascnt{example}{theorem}
\newtheorem{example}[example]{Example}
\aliascntresetthe{example}

\crefname{theorem}{theorem}{theorems}
\crefname{proposition}{proposition}{propositions}
\crefname{lemma}{lemma}{lemmas}
\crefname{corollary}{corollary}{corollaries}
\crefname{definition}{definition}{definitions}
\crefname{remark}{remark}{remarks}
\crefname{example}{example}{examples}

\newcommand{\R}{\mathbb{R}}

\newcommand{\one}{\mathbf{1}}
\newcommand{\e}{\mathrm{e}}
\newcommand{\diag}{\operatorname{Diag}}
\newcommand{\softmax}{\operatorname{softmax}}
\newcommand{\smx}[1]{\softmax_{\eta}\!\left(#1\right)}
\newcommand{\rank}{\operatorname{rank}}

\newcommand{\poly}{\operatorname{poly}}

\newcommand{\Rcont}{R_{\mathrm{cont}}}
\newcommand{\Rstar}{R^*_{\mathrm{cont}}}
\newcommand{\KL}{D_{\mathrm{KL}}}
\newcommand{\dd}{\,\mathrm{d}}

\title{Planning Against Learning in Rank-1 Games}
\author{William Overman\thanks{The main reduction and portions of the writing of this paper were developed through discussion with and assistance from OpenAI’s ChatGPT (GPT-5.6 Sol). The author directed the research, verified the proof, and takes responsibility for the paper’s contents.}\\
Graduate School of Business, Stanford University\\
\texttt{wpo@stanford.edu}}
\date{}

\begin{document}
\maketitle

\begin{abstract}
Learning algorithms are often used to make decisions in repeated multi-agent environments. When another player understands how a learner adapts from past experience, that player can plan strategically across rounds to influence the learner's future behavior. Recent work shows that optimizing against Replicator Dynamics, the continuous-time analogue of Multiplicative Weights Update, is tractable in zero-sum games but can be hard in unrestricted general-sum games. We study the first structured class beyond zero sum: bimatrix games satisfying $\rank(A+B)=1$, for which Nash equilibria can be computed in polynomial time.

Our main result shows that this equilibrium tractability does not extend to planning against learning dynamics. Unless $\mathsf{P}=\mathsf{NP}$, approximating the optimizer's optimal continuous-time reward within a fixed additive constant is NP-hard even when $\rank(A+B)=1$, the learner starts from the uniform state, and the optimizer is restricted to constant strategies. The hardness persists for bounded payoff matrices and polynomially bounded horizons. We complement this result with structural characterizations of several tractable special cases. Thus rank-one games already separate efficient equilibrium computation from strategic planning against a learning opponent.
\end{abstract}

\section{Introduction}

Online learning algorithms increasingly interact with strategic agents that understand how the learner updates. In such settings, an optimizer may sacrifice immediate utility in order to steer the learner toward future behavior that is more favorable to the optimizer. This optimizer-versus-learner viewpoint has been studied in auctions, contracts, persuasion, and repeated normal-form games; see, for example, \citet{deng2019strategizing} and the discussion in \citet{assos2024maximizing}.

\citet{assos2024maximizing} formalized a basic repeated bimatrix model with two players. The row player, called the \emph{optimizer}, knows the column player's payoff matrix and learning algorithm and chooses a time-dependent strategy to maximize cumulative utility. The column player, called the \emph{learner}, follows Replicator Dynamics in continuous time, or Multiplicative Weights Update in discrete time. In zero-sum games $(A,-A)$, they showed that the optimizer's continuous-time reward depends only on its aggregate play, so an optimal strategy is constant and can be computed by convex optimization. For general-sum games, they established hardness against a history-best-response learner; subsequent work proved an $\Omega(T)$ hardness gap against MWU in unrestricted games \citep{assos2025intractability}.

The continuous-time general-sum problem against Replicator Dynamics remained open. In particular, \citet{assos2024maximizing} suggested studying the first nonzero level of the standard rank hierarchy,
\[
    \rank(A+B)=1.
\]
This class is notable because exact Nash equilibrium computation is polynomial-time solvable \citep{adsul2021fast}, and tailored learning dynamics can efficiently reach approximate Nash equilibria \citep{patris2024learning}. The question is whether this structural tractability survives when one player optimally plans against the other's learning dynamics.

We show that the apparent tractability of rank-one games does not extend to strategic planning against a learner. Optimizing against Replicator Dynamics remains NP-hard even when $\rank(A+B)=1$, including when the optimizer is restricted to constant strategies. Nevertheless, the rank-one structure exposes several tractable subclasses: two-action learners admit a strong ordering characterization, while conservative instances reduce to static optimization.

\subsection{Our contributions}

Our main result is computational hardness.

\begin{theorem}[Informal]
Unless $\mathsf{P}=\mathsf{NP}$, no polynomial-time algorithm can approximate the optimal continuous-time value within additive error $1/10$ on every rational game satisfying $\rank(A+B)=1$. The result remains true if the optimizer is restricted to constant strategies. It also remains true for matrices $A,B\in[0,1]^{n\times m}$ after allowing a polynomially bounded horizon.
\end{theorem}

The reduction is from \textsc{Clique}. For each graph edge, we create four learner actions whose signed logits synthesize a smooth approximation to the Motzkin--Straus quadratic program. The optimizer's best continuous-time reward then reveals the graph's clique number up to a constant additive error. Notably, the reduction does not rely on complicated temporal switching: hardness already holds for choosing the best constant mixed strategy.

We complement this negative result with structural results that explain what survives from the zero-sum case. When $A+B=uv^\top$, the optimizer's reward separates into the familiar zero-sum endpoint term and one additional scalar term that can depend on the order of play. We characterize instances in which this additional term is path-independent, so constant play remains optimal. When the learner has two actions, optimizer actions admit a global block ordering under a simple sign condition, and every $2\times2$ instance has an optimal strategy with at most one switch. A minimal \emph{train-then-harvest} game shows that this dynamic advantage is real: one switch can strictly outperform every constant strategy, with an exact closed-form optimum.

\subsection{Related work}

\paragraph{Strategizing against learning agents.}
A modern line of work studies a strategic optimizer that knows, and plans against, another player's online learning rule. \citet{braverman2018selling} initiated this perspective in repeated auctions, showing that a seller can exploit a mean-based buyer to obtain revenue close to welfare. \citet{deng2019strategizing} abstracted the interaction to general normal-form games: an optimizer can obtain an $\Omega(T)$ advantage over the Stackelberg benchmark against mean-based learners, whereas swap-regret substantially limits this advantage. \citet{mansour2022bayesian} developed the corresponding theory in Bayesian games, and related optimizer--learner interactions have since been studied with multiple buyers \citep{cai2023multiple}, learning agents in contract design \citep{guruganesh2024contracting}, and the reward--regret tradeoffs faced by the learner itself \citep{brown2023learning}. For particular dynamics, \citet{guo2023hedge} solve the finite-horizon optimization problem against Hedge in $2\times2$ zero-sum games.

Our closest antecedent is \citet{assos2024maximizing}. They formulate the continuous-time optimizer value against Replicator Dynamics, prove that zero-sum instances reduce to endpoint convex optimization, compare the continuous and discrete dynamics, and establish a first general-sum hardness result against a history-best-response learner. Their conclusion explicitly proposes $\rank(A+B)=1$ as the first structured general-sum case in which to understand the optimizer's value and policy. \citet{assos2025intractability} subsequently strengthen the unrestricted hardness picture to an $\Omega(T)$ gap against MWU. We instead resolve the proposed rank-one continuous-time direction: the restriction yields substantial geometric structure, but does not prevent NP-hardness when the learner's action set is unrestricted.

\paragraph{The rank hierarchy for bimatrix games.}
The rank used here is $\rank(A+B)$, introduced by \citet{kannan2010games} as a hierarchy interpolating between zero-sum and general bimatrix games. This differs from requiring the individual payoff matrices $A$ and $B$ to have low rank, as in earlier small-support results \citep{lipton2003playing}. \citet{kannan2010games} gave polynomial-time approximation algorithms for every fixed rank and showed that rank-one games can already have arbitrarily many connected components of equilibria. \citet{theobald2009enumerating} showed that not every rank-one equilibrium lies on a Lemke--Howson path and gave a parametric algorithm for enumerating the equilibria of nondegenerate rank-one games. Rank-one games may in fact have exponentially many Nash equilibria \citep{vonstengel2012exponential}.

The main positive equilibrium result exploits the representation $A+B=ab^\top$. An equilibrium of $(A,-A+ab^\top)$ can be located through the parameterized zero-sum family
\[
    \bigl(A-\one\lambda b^\top,\,-A+\one\lambda b^\top\bigr)
\]
together with the scalar consistency condition $x^\top a=\lambda$. \citet{adsul2011homeomorphism} used the geometry of the fully labeled set to obtain the first polynomial-time exact algorithm, and \citet{adsul2021fast} developed the equilibrium-path and parametric-linear-programming structure into faster algorithms that also cover degenerate games. Related rank-based methods extend beyond simplex strategy spaces to bilinear games \citep{garg2011bilinear}. By contrast, exact equilibrium computation is PPAD-hard for every fixed rank at least three, while the rank-two case remains unresolved \citep{mehta2018constant}.

The equilibrium-learning literature gives another positive side of the hierarchy. Broader work studies last-iterate and optimistic learning beyond zero-sum games \citep{anagnostides2022lastiterate,anagnostides2022omd}. Building directly on the parameterized zero-sum representation and the equilibrium landscape above, \citet{patris2024learning} modify optimistic mirror descent to reach an $\epsilon$-approximate Nash equilibrium in $O(\epsilon^{-2}\log(1/\epsilon))$ iterations, despite rank-one games not generally satisfying the Minty variational-inequality condition. The formal resemblance to our decomposition is useful but limited: those works vary a scalar parameter to satisfy an equilibrium consistency condition, whereas our optimizer endogenously moves the learner's state and maximizes a path-dependent finite-horizon functional. Thus, their tractability results do not directly control the transient optimization problem studied here, and our hardness theorem shows that equilibrium tractability does not transfer to strategic planning against the learning dynamics.

\paragraph{Motivational connections to AI safety.}
Game-theoretic models have long been used to study alignment and oversight. The Off-Switch Game and cooperative inverse reinforcement learning model human--AI interaction as cooperative games in which uncertainty or asymmetric information about reward plays a central role \citep{hadfieldmenell2016offswitch,hadfieldmenell2016cirl}. More recent work considers broader oversight interactions that need not be fully cooperative, including Markov-game formulations of autonomy, deferral, trust, and intervention \citep{overman2025oversight}. Learned optimization, reward tampering, and AI control provide complementary motivations for studying strategic systems that may anticipate how feedback or oversight adapts over time \citep{hubinger2019risks,everitt2021reward,greenblatt2024control}. These connections motivate optimizer--learner interactions, but none of them implies the rank-one structure studied here.

Our hardness proof uses the classical Motzkin--Straus characterization of maximum clique as quadratic optimization over a simplex \citep{motzkin1965maxima}. The signed four-column hyperbolic edge gadget is tailored to the rank-one identity $A+B=uv^\top$.

\section{Model and preliminaries}

For $d\in\mathbb{N}$, let
\[
    \Delta_d=\left\{x\in\R^d_+:\one^\top x=1\right\}
\]
denote the probability simplex. The optimizer has $n$ pure actions and payoff matrix $A\in\R^{n\times m}$; the learner has $m$ pure actions and payoff matrix $B\in\R^{n\times m}$. We write $x(t)\in\Delta_n$ and $y(t)\in\Delta_m$ for their mixed strategies.

Fix a horizon $T>0$, learning rate $\eta>0$, and initial historical-payoff vector $h^0\in\R^m$. An admissible optimizer strategy is a measurable map
\[
    x:[0,T]\to\Delta_n.
\]
Its cumulative play and the learner's historical-payoff state are
\begin{equation}
    z_x(t)=\int_0^t x(s)\dd s,
    \qquad
    h_x(t)=h^0+B^\top z_x(t).
    \label{eq:state}
\end{equation}
The learner follows Replicator Dynamics in its exponential-weights representation:
\begin{equation}
    y_{x,j}(t)
    =\frac{\exp(\eta h_{x,j}(t))}
           {\sum_{k=1}^m\exp(\eta h_{x,k}(t))},
    \qquad j\in[m].
    \label{eq:rd}
\end{equation}
This is equivalent to the usual replicator ODE whenever the learner starts in the interior of the simplex. We write $y_x(t)=\softmax_{\eta}(h_x(t))$ and use this notation uniformly below.

The optimizer's cumulative reward is
\begin{equation}
    \Rcont(x,h^0,T,A,B)
    =\int_0^T x(t)^\top A y_x(t)\dd t,
    \label{eq:rcont}
\end{equation}
and the optimal value is
\begin{equation}
    \Rstar(h^0,T,A,B)
    =\sup_{x:[0,T]\to\Delta_n}\Rcont(x,h^0,T,A,B).
    \label{eq:rstar}
\end{equation}
We use the supremum convention throughout; the controls constructed in our positive results and hardness reduction attain the corresponding bounds.

\begin{definition}[Rank]
The rank of a bimatrix game $(A,B)$ is $\rank(A+B)$. A rank-one game admits a factorization
\[
    A+B=uv^\top
\]
for vectors $u\in\R^n$ and $v\in\R^m$.
\end{definition}

We use the scaled log-partition function
\begin{equation}
    \Phi_\eta(h)=\frac1\eta\log\left(\sum_{j=1}^m\e^{\eta h_j}\right).
    \label{eq:phi}
\end{equation}
Its gradient is the softmax map:
\[
    \nabla\Phi_\eta(h)=\smx{h}.
\]

\section{The rank-one reduction}
\label{sec:rankone-reduction}

\begin{proposition}[Rank-one decomposition]
\label{prop:decomposition}
Suppose $A+B=uv^\top$. For every admissible optimizer strategy $x$,
\begin{equation}
\boxed{
    \Rcont(x,h^0,T,A,B)
    =\Phi_\eta(h^0)-\Phi_\eta(h_x(T))
     +\int_0^T (u^\top x(t))(v^\top y_x(t))\dd t.}
    \label{eq:rankone-decomp}
\end{equation}
\end{proposition}

\begin{proof}
Since $A=-B+uv^\top$,
\[
    x^\top Ay=-x^\top By+(u^\top x)(v^\top y).
\]
By \cref{eq:state}, $\dot h_x(t)=B^\top x(t)$ almost everywhere. Hence
\[
    \frac{\dd}{\dd t}\Phi_\eta(h_x(t))
    =\nabla\Phi_\eta(h_x(t))^\top B^\top x(t)
    =x(t)^\top B y_x(t).
\]
Integrating gives \cref{eq:rankone-decomp}.
\end{proof}

Let
\begin{equation}
    \mathcal{D}_T
    =\left\{z\in\R^n_+:\one^\top z\le T\right\},
\end{equation}
and define
\begin{equation}
    y(z)=\smx{h^0+B^\top z},
    \qquad
    g(z)=v^\top y(z).
    \label{eq:gdef}
\end{equation}
The cumulative-play path $z_x(t)=\int_0^t x(s)\dd s$ is coordinatewise nondecreasing and satisfies $\one^\top z_x(t)=t$. Equation \eqref{eq:rankone-decomp} can therefore be written as
\begin{equation}
    \Rcont(x,h^0,T,A,B)
    =\Phi_\eta(h^0)-\Phi_\eta(h^0+B^\top z_x(T))
     +\int_{\gamma_x} g(z)u^\top\dd z,
    \label{eq:path-form}
\end{equation}
where $\gamma_x$ is the path traced by $z_x$. Equivalently, the additional term is the line integral of
\begin{equation}
    \omega=g(z)u^\top\dd z.
    \label{eq:oneform}
\end{equation}

In a zero-sum game, the final term is absent and reward depends only on aggregate play. Rank one adds a single scalar term through which two schedules with the same aggregate strategy can earn different rewards. This is the only source of a temporal advantage in the rank-one problem.

\paragraph{Tractable special cases.}
The appendix develops the geometry of this additional term. \Cref{app:curvature} derives an adjacent-swap identity that measures the value of reordering two optimizer actions. \Cref{app:pathind} characterizes the instances in which the term is path-independent; in those cases, constant play remains optimal and the problem reduces to endpoint optimization.

\section{A minimal dynamic advantage}
\label{sec:minimal-dynamic}

The following $2\times2$ example shows that the zero-sum conclusion that constant play is optimal already fails at rank one.

\begin{example}[Train-then-harvest]
\label{ex:train-harvest}
Let
\begin{equation}
    A=\begin{pmatrix}0&0\\0&1\end{pmatrix},
    \qquad
    B=\begin{pmatrix}0&1\\0&0\end{pmatrix}.
    \label{eq:train-matrices}
\end{equation}
Then
\[
    A+B=\one\begin{pmatrix}0&1\end{pmatrix}
\]
has rank one. Row one provides no immediate reward but shifts the learner toward column two; row two then earns reward precisely when column two is played. Starting from $h^0=0$, if the optimizer spends $\tau$ time units training on row one and then switches to row two, its reward is
\begin{equation}
    R(\tau)=(T-\tau)\sigma(\eta\tau),
    \qquad
    \sigma(s)=\frac{1}{1+\e^{-s}}.
    \label{eq:train-reward}
\end{equation}
For any fixed total training time $\tau$, placing all training before harvesting is optimal; the proof is given in \cref{app:train-harvest}.
\end{example}

The principal Lambert $W$ function is characterized by $W(x)\e^{W(x)}=x$ for $x\ge0$.

\begin{proposition}[Closed-form optimum]
\label{prop:lambert}
For the game in \cref{ex:train-harvest},
\begin{equation}
    \tau^*=0,
    \qquad
    \Rstar(0,T,A,B)=\frac T2,
    \qquad\text{if }\eta T\le2.
\end{equation}
If $\eta T>2$, the unique optimal switch time and value are
\begin{equation}
\boxed{
    \tau^*
    =T-\frac{1+W(\e^{\eta T-1})}{\eta},
    \qquad
    \Rstar(0,T,A,B)
    =\frac1\eta W(\e^{\eta T-1}).}
    \label{eq:lambert}
\end{equation}
\end{proposition}

The advantage can be linear in the horizon under the natural scaling $\eta=\Theta(1/T)$.

\begin{proposition}[Linear separation from constant play]
\label{prop:linear-gap}
For the game in \cref{ex:train-harvest}, fix $L\in(2,4]$ and set $\eta=L/T$. Then
\begin{equation}
    \max_{p\in[0,1]}\Rcont(x(t)\equiv(p,1-p),0,T,A,B)=\frac T2,
\end{equation}
while
\begin{equation}
\boxed{
    \Rstar(0,T,A,B)-\max_{\text{constant }x}\Rcont(x,0,T,A,B)
    =\left(\frac{W(\e^{L-1})}{L}-\frac12\right)T>0.}
    \label{eq:linear-gap}
\end{equation}
\end{proposition}

The example is part of a broader two-action theory. \Cref{app:two-actions} shows that pairwise ordering preferences are state-independent when the learner has two actions. Under a common-sign condition on the rank-one factor $u$, this yields a global contiguous-block ordering, while every $2\times2$ instance has an optimal strategy with at most one switch. Proofs of \cref{prop:lambert,prop:linear-gap} are also deferred to the appendix.

\section{NP-hardness for general rank-one games}
\label{sec:hardness}

We now show that the rank-one problem is computationally intractable once the learner has an unrestricted number of actions. The reduction is from \textsc{Clique}, which is NP-complete \citep{karp1972reducibility}.

\paragraph{Proof idea.}
For each graph edge $\{p,q\}$, the \emph{hyperbolic edge gadget} introduces four learner logits, $\pm\rho(z_p+z_q)$ and $\pm\rho(z_p-z_q)$, with signs $+,+,-,-$ in the rank-one vector $v$. The hyperbolic addition identities make these four terms synthesize $\sinh(\rho z_p)\sinh(\rho z_q)$ in the numerator of the scalar observable $v^\top y$, while the corresponding softmax partition function contributes $\cosh(\rho z_p)\cosh(\rho z_q)$ in the denominator. For small $\rho$, the ratio uniformly approximates the edge monomial $z_pz_q$; summing the gadgets therefore expresses a graph quadratic through a single rank-one observable.

The Motzkin--Straus theorem identifies the maximum of $z^\top W_Gz$ over the simplex as $1-1/\omega(G)$, where $\omega(G)$ is the clique number. Homogeneity further implies that the dynamic integral of this quadratic is maximized by a constant strategy. The remaining work is to make the softmax approximation uniform, control the log-partition endpoint term, and preserve a constant gap under rational encoding.

\subsection{Statement of the main result}

We measure input size by the binary encoding length of rational $A,B,h^0,T,$ and $\eta$. An additive-$\epsilon$ value algorithm returns $\widehat R$ satisfying $|\widehat R-\Rstar|\le\epsilon$.

\begin{theorem}[Constant-additive hardness]
\label{thm:hardness}
Unless $\mathsf{P}=\mathsf{NP}$, there is no polynomial-time algorithm that, for every rational rank-one game, returns $\widehat R$ such that
\begin{equation}
    \left|\widehat R-\Rstar(0,1,A,B)\right|\le\frac1{10},
    \label{eq:additive-hard}
\end{equation}
when the learner follows Replicator Dynamics with $\eta=1$.

The hardness holds even for the restricted problem in which the optimizer must use a constant strategy. The constructed instance has $m=O(n^2)$ learner actions and polynomial-bit-complexity payoffs.
\end{theorem}

We first prove the theorem with polynomially bounded but not normalized payoffs, then normalize to $[0,1]$ in \cref{subsec:normalization}.

\subsection{The reduction}

Let $(G,k)$ be an instance of \textsc{Clique}, where $G=(V,E)$ has $n=|V|$ vertices and $M=|E|\ge1$ edges. We may restrict to $3\le k\le n$ and graphs with at least one edge without changing NP-hardness.

Set
\begin{equation}
    \rho=\frac1{100n},
    \qquad
    K=n^2,
    \qquad
    \Lambda=\frac{2MK}{\rho^2}=20000Mn^4.
    \label{eq:parameters}
\end{equation}
The optimizer has one action per vertex. For each edge $e=\{p,q\}\in E$, create four learner actions whose columns in $B$ and coordinates in $v$ are
\begin{equation}
\begin{array}{c|c}
\text{column of }B&\text{coordinate of }v\\
\midrule
\rho(e_p+e_q)&+\Lambda\\
-\rho(e_p+e_q)&+\Lambda\\
\rho(e_p-e_q)&-\Lambda\\
-\rho(e_p-e_q)&-\Lambda.
\end{array}
\label{eq:gadget}
\end{equation}
Thus $m=4M$. Finally define
\begin{equation}
    A=\one v^\top-B.
    \label{eq:Aconstruct}
\end{equation}
Then $A+B=\one v^\top$ has rank exactly one.

Because $u=\one$ and every optimizer strategy satisfies $\one^\top x(t)=1$, \cref{prop:decomposition} gives
\begin{equation}
    \Rcont(x,0,1,A,B)
    =\Phi_1(0)-\Phi_1(B^\top z(1))
     +\int_0^1 g(z(t))\dd t,
    \label{eq:hard-decomp}
\end{equation}
where
\[
    g(z)=v^\top\softmax_{1}\!\left(B^\top z\right).
\]

\subsection{The hyperbolic edge gadget approximates the clique quadratic}

Let $W_G$ denote the adjacency matrix of $G$, and define
\begin{equation}
    q_G(z)=z^\top W_Gz=2\sum_{\{p,q\}\in E}z_pz_q.
    \label{eq:qG}
\end{equation}

\begin{lemma}[Hyperbolic edge gadget: exact representation]
\label{lem:hyperbolic}
For every $z\in\mathcal D_1$,
\begin{equation}
\boxed{
    g(z)=\Lambda
    \frac{\displaystyle\sum_{\{p,q\}\in E}
      \sinh(\rho z_p)\sinh(\rho z_q)}
    {\displaystyle\sum_{\{p,q\}\in E}
      \cosh(\rho z_p)\cosh(\rho z_q)}.}
    \label{eq:hyperbolic}
\end{equation}
\end{lemma}

\begin{proof}
For one edge $\{p,q\}$, the four logits in \cref{eq:gadget} are
\[
    \pm\rho(z_p+z_q),\qquad \pm\rho(z_p-z_q).
\]
Their contribution to the numerator of $v^\top\softmax_{1}\!\left(B^\top z\right)$ before normalization is
\begin{align*}
    2\Lambda\cosh(\rho(z_p+z_q))
    -2\Lambda\cosh(\rho(z_p-z_q))
    =4\Lambda\sinh(\rho z_p)\sinh(\rho z_q),
\end{align*}
while their contribution to the softmax denominator is
\begin{align*}
    2\cosh(\rho(z_p+z_q))+2\cosh(\rho(z_p-z_q))
    =4\cosh(\rho z_p)\cosh(\rho z_q).
\end{align*}
Summing over edges and canceling the factor four proves the claim.
\end{proof}

\begin{lemma}[Uniform quadratic approximation]
\label{lem:approx}
For every $z\in\mathcal D_1$,
\begin{equation}
    \frac{Kq_G(z)}{(1+\rho^2)^2}
    \le g(z)\le
    Kq_G(z)(1+\rho^2)^2.
    \label{eq:multiplicative-approx}
\end{equation}
Consequently,
\begin{equation}
\boxed{
    |g(z)-Kq_G(z)|\le3K\rho^2.}
    \label{eq:additive-approx}
\end{equation}
\end{lemma}

\begin{proof}
For $0\le r\le1$,
\begin{equation}
    1\le\frac{\sinh r}{r}\le1+r^2,
    \qquad
    1\le\cosh r\le1+r^2.
    \label{eq:elementary-hyperbolic}
\end{equation}
Indeed, for $0\le r\le1$,
\[
    \frac{\sinh r}{r}
    =1+\sum_{\ell\ge1}\frac{r^{2\ell}}{(2\ell+1)!}
    \le1+r^2\sum_{\ell\ge1}\frac1{(2\ell+1)!}
    <1+r^2\quad\text{for }r>0,
\]
with equality in the displayed non-strict bound at $r=0$; the same argument gives
\[
    \cosh r
    =1+\sum_{\ell\ge1}\frac{r^{2\ell}}{(2\ell)!}
    \le1+r^2\sum_{\ell\ge1}\frac1{(2\ell)!}
    <1+r^2\quad\text{for }r>0.
\] Since $0\le z_i\le1$ and $\rho<1$, the numerator in \cref{eq:hyperbolic} lies between
\[
    \rho^2\sum_{\{p,q\}\in E}z_pz_q
    \quad\text{and}\quad
    \rho^2(1+\rho^2)^2\sum_{\{p,q\}\in E}z_pz_q,
\]
and the denominator lies between $M$ and $M(1+\rho^2)^2$. Substituting $\Lambda=2MK/\rho^2$ and using \cref{eq:qG} proves \cref{eq:multiplicative-approx}.

Since $0\le q_G(z)\le(\one^\top z)^2\le1$ and
\[
    (1+\rho^2)^2-1\le3\rho^2,
\]
the additive bound follows.
\end{proof}

\subsection{The endpoint term is negligible}

\begin{lemma}[Endpoint bound]
\label{lem:endpoint}
For every $z\in\mathcal D_1$,
\begin{equation}
    -2\rho^2
    \le\Phi_1(0)-\Phi_1(B^\top z)
    \le0.
    \label{eq:endpoint-bound}
\end{equation}
\end{lemma}

\begin{proof}
The total exponential mass of the four columns associated with edge $\{p,q\}$ is
\[
    4\cosh(\rho z_p)\cosh(\rho z_q).
\]
Since $m=4M$,
\[
    \Phi_1(0)-\Phi_1(B^\top z)
    =\log\frac{M}{\sum_{\{p,q\}\in E}\cosh(\rho z_p)\cosh(\rho z_q)}.
\]
The denominator lies in $[M,M(1+\rho^2)^2]$. Therefore the expression lies in
\[
    [-2\log(1+\rho^2),0]\subseteq[-2\rho^2,0].
\]
\end{proof}

\subsection{The dynamic quadratic collapses to Motzkin--Straus}

Let
\begin{equation}
    q_G^*=\max_{w\in\Delta_n}q_G(w).
\end{equation}

\begin{lemma}[Dynamic homogeneous quadratic]
\label{lem:dynamic-q}
\begin{equation}
\boxed{
    \sup_{x:[0,1]\to\Delta_n}
      \int_0^1q_G(z_x(t))\dd t
    =\frac{q_G^*}{3}.}
    \label{eq:dynamic-q}
\end{equation}
Moreover, the supremum is attained by a constant strategy.
\end{lemma}

\begin{proof}
For $t>0$, $z_x(t)/t\in\Delta_n$. Since $q_G$ is homogeneous of degree two,
\[
    q_G(z_x(t))
    =t^2q_G(z_x(t)/t)
    \le t^2q_G^*.
\]
Integrating gives the upper bound $q_G^*/3$. If $w^*$ maximizes $q_G$ on the simplex, the constant strategy $x(t)=w^*$ satisfies $z_x(t)=tw^*$ and attains equality.
\end{proof}

The Motzkin--Straus theorem states
\begin{equation}
\boxed{
    q_G^*=1-\frac1{\omega(G)},}
    \label{eq:motzkin-straus}
\end{equation}
where $\omega(G)$ is the clique number.

Combining the previous lemmas yields the quantitative reduction.

\begin{proposition}[Value approximation]
\label{prop:value-approx}
For the constructed rank-one game,
\begin{equation}
\boxed{
    \left|
    \Rstar(0,1,A,B)
    -\frac K3\left(1-\frac1{\omega(G)}\right)
    \right|
    \le(3K+2)\rho^2.}
    \label{eq:value-approx}
\end{equation}
The same bound holds if $\Rstar$ is replaced by the optimal value over constant optimizer strategies.
\end{proposition}

\begin{proof}
By \cref{lem:approx,lem:endpoint,lem:dynamic-q}, every optimizer strategy has reward at most
\[
    \frac{Kq_G^*}{3}+3K\rho^2.
\]
The constant strategy $w^*$ from \cref{lem:dynamic-q} has reward at least
\[
    \frac{Kq_G^*}{3}-(3K+2)\rho^2.
\]
Use \cref{eq:motzkin-straus}. Since the lower-bound strategy is constant, the same sandwich applies to the constant-strategy optimum.
\end{proof}

\subsection{Completing the reduction}

With $K=n^2$ and $\rho=1/(100n)$,
\begin{equation}
    (3K+2)\rho^2
    =\frac{3n^2+2}{10000n^2}
    <\frac1{2000}
    \label{eq:error-small}
\end{equation}
for $n\ge3$.

If $\omega(G)\ge k$, then
\[
    \Rstar\ge\frac{n^2}{3}\left(1-\frac1k\right)-\frac1{2000}.
\]
If $\omega(G)\le k-1$, then
\[
    \Rstar\le\frac{n^2}{3}\left(1-\frac1{k-1}\right)+\frac1{2000}.
\]
The gap between the ideal values is
\begin{equation}
    \frac{n^2}{3k(k-1)}\ge\frac13,
    \label{eq:clique-gap}
\end{equation}
because $k\le n$. After approximation errors, the gap remains larger than $0.332$. Allowing an additional error of $1/10$ on each side still leaves a separation larger than $0.132$, so comparison with the midpoint of the two displayed thresholds decides \textsc{Clique}. This proves \cref{thm:hardness}. The same argument applies to the optimal constant-strategy value by \cref{prop:value-approx}.

\subsection{Normalization to bounded nonnegative payoffs}
\label{subsec:normalization}

The preceding construction uses signed and polynomially large entries. We now show that this is not essential.

Let $J=\one_n\one_m^\top$ and define
\begin{equation}
    B^+=B+\rho J,
    \qquad
    v^+=v+(\Lambda+2\rho)\one_m,
    \qquad
    A^+=\one_n(v^+)^\top-B^+.
    \label{eq:positive-shift}
\end{equation}
Then
\[
    A^++B^+=\one_n(v^+)^\top
\]
has rank one. Also, $B^+\in[0,2\rho]^{n\times m}$. On the columns with original $v_j=-\Lambda$, we have $v_j^+=2\rho$ and $A^+_{ij}=2\rho-B^+_{ij}\in[0,2\rho]$. On the columns with original $v_j=+\Lambda$, we have $A^+_{ij}\in[2\Lambda,2\Lambda+2\rho]$. Thus $A^+,B^+$ are nonnegative.

Adding $\rho J$ to $B$ adds the common quantity $\rho t$ to every learner logit at time $t$, so the learner strategy is unchanged. Moreover,
\begin{equation}
    A^+=A+(\Lambda+\rho)J,
\end{equation}
so every optimizer strategy receives the known additive reward $\Lambda+\rho$ over the unit horizon.

Let
\begin{equation}
    C=2\Lambda+2,
    \qquad
    \bar A=A^+/C,
    \qquad
    \bar B=B^+/C,
    \qquad
    \bar T=C.
    \label{eq:time-scale}
\end{equation}
Then $\bar A,\bar B\in[0,1]^{n\times m}$ and $\rank(\bar A+\bar B)=1$. The time change $\bar x(t)=x(t/C)$ gives
\begin{equation}
    \Rstar(0,C,\bar A,\bar B)
    =\Rstar(0,1,A^+,B^+)
    =\Rstar(0,1,A,B)+\Lambda+\rho.
    \label{eq:scaling-value}
\end{equation}
The final shift is known from the constructed instance, so it can be subtracted from any value estimate. Since $C=O(Mn^4)$, the new horizon is polynomially bounded and has polynomial encoding length.

\begin{corollary}[Hardness with bounded payoffs]
\label{cor:bounded-hardness}
Unless $\mathsf{P}=\mathsf{NP}$, additive-$1/10$ approximation of $\Rstar(0,T,A,B)$ is NP-hard for rank-one games satisfying
\[
    A,B\in[0,1]^{n\times m},
    \qquad
    \eta=1,
    \qquad
    T=\poly(n,m).
\]
The result remains true for the constant-strategy optimum. In particular, the bounded-payoff problem admits no fully polynomial-time approximation scheme unless $\mathsf{P}=\mathsf{NP}$.
\end{corollary}
\begin{proof}
The additive-hardness claim follows from the preceding transformation and the known shift in \cref{eq:scaling-value}. For the final statement, note that every bounded-payoff instance has value at most $T$. On the constructed family, an FPTAS run with relative accuracy $\epsilon=1/(20T)$ would therefore have additive error at most $1/20$; moreover, $1/\epsilon=20T$ is polynomial in the size of the original \textsc{Clique} instance. This would contradict \cref{thm:hardness}.
\end{proof}

\section{Discussion and open questions}

Our results give the following initial complexity picture for optimizing against Replicator Dynamics:
\begin{center}
\begin{tabular}{@{}p{0.41\linewidth}p{0.53\linewidth}@{}}
\toprule
Game class & Status\\
\midrule
$A+B=0$ & Constant optimum; convex optimization \citep{assos2024maximizing}\\
$\rank(A+B)=1$, conservative one-form & Constant optimum; endpoint optimization\\
$\rank(A+B)=1$, two learner actions & Ordering under a sign condition; $2\times2$ has one switch\\
$\rank(A+B)=1$, unrestricted learner actions & NP-hard to approximate\\
\bottomrule
\end{tabular}
\end{center}

The reduction shows that the static rank of $A+B$ is not by itself the right tractability parameter. A natural complementary quantity is the effective dimension of the learner's log-odds state,
\begin{equation}
    d_{\mathrm{eff}}
    =\rank\left(B\left(I_m-\frac1m\one\one^\top\right)\right).
    \label{eq:deff}
\end{equation}
This is the dimension of the centered learner-payoff state that the optimizer can steer. Our hardness construction has growing $d_{\mathrm{eff}}$, whereas $m=2$ forces $d_{\mathrm{eff}}\le1$. Dimension alone, however, does not eliminate temporal advantage: the train-then-harvest example already has a one-dimensional learner state. The stronger condition in \cref{thm:pathind} is geometric---the variation in the scalar rank-one reward field must align with the vector $u$ so that only aggregate play matters. \Cref{cor:onedim} gives the concrete sufficient condition $BP=uw^\top$. This suggests that both the dimension of the steerable learner state and its alignment with the rank-one direction may be relevant parameters for tractability.

Several questions remain open.
\begin{enumerate}[leftmargin=2em]
    \item Is the rank-one problem fixed-parameter tractable in $m$ or $d_{\mathrm{eff}}$? In particular, is there a polynomial-time algorithm for exactly three learner actions?
    \item Can the global ordering theorem be extended beyond $m=2$ under total positivity, monotone likelihood-ratio, or sign-regularity conditions on $B$ and $v$?
    \item What is the complexity of rank-one optimization against discrete-time MWU? The discrete problem includes an additional Bregman-movement term, recorded in \cref{app:discrete}.
    \item Can one approximate the general rank-one value by exploiting the rank-two curvature field, perhaps under bounded variation or small $\eta T$?
    \item Which parts of the structure survive for learners with discounted or forgetful state, for example $\dot h=B^\top x-\lambda h$? Such dynamics remove the exact cumulative-history representation and may change both the path-independence criterion and the optimizer's ability to profit from long-lived state manipulation.
\end{enumerate}

\paragraph{Implications for AI oversight.}
A natural interpretation of the optimizer--learner interaction is a capable agent facing an adaptive safety mechanism: for example, a trusted monitor that learns when to intervene, a reward model updated from observed trajectories, or a control protocol that grants greater autonomy as the agent accumulates a history of apparently safe behavior. In such systems, an agent may care not only about whether its current action passes oversight, but also about how that action changes the oversight it will face later. Our train-then-harvest example gives a minimal version of this familiar ``build trust, then exploit'' pattern: the optimizer first behaves so as to move the learner into a favorable state and only later takes advantage of that state. The strategy is simple and explicit: an agent only needs to find one profitable way of shaping oversight, whereas an auditor trying to certify robustness must rule out all strategies that could do better. Our hardness result makes this asymmetry concrete: even approximating how much reward the best optimizer can extract is NP-hard in general, and the hardness persists even when the optimizer is restricted to constant play. Our positive results identify the complementary possibility: for some feedback structures, the order of behavior carries no additional value, so an agent cannot benefit from appearing one way early merely to induce a more favorable response later. Viewed through the lens of adaptive monitoring or AI control, the central question is therefore not only whether an overseer performs well on the behavior it has seen, but whether its update rule itself creates a strategically exploitable state.

\paragraph{What governs tractability.}
Rank-one games occupy a genuinely intermediate regime. The zero-sum endpoint calculation survives, augmented by a single path-dependent term whose local curvature has rank at most two; two-action learners retain strong ordering structure; yet the unrestricted planning problem is NP-hard. The broader lesson concerns parameters. The rank of $A+B$ is a central measure of departure from zero sum in equilibrium computation, but it is insufficient to control the complexity of transient planning against a learner. Our results instead point to the dimension and geometry of the learner state that the optimizer can steer as the more relevant candidates. Static equilibrium structure, however strong, does not automatically transfer to transient control of learning dynamics: the two problems already come apart at the first nonzero level of the rank hierarchy.

\bibliographystyle{plainnat}
\bibliography{references}

@inproceedings{assos2024maximizing,
  title={Maximizing Utility in Multi-Agent Environments by Anticipating the Behavior of Other Learners},
  author={Assos, Angelos and Dagan, Yuval and Daskalakis, Constantinos},
  booktitle={Advances in Neural Information Processing Systems},
  volume={37},
  year={2024}
}

@inproceedings{assos2025intractability,
  title={Computational Intractability of Strategizing against Online Learners},
  author={Assos, Angelos and Dagan, Yuval and Rajaraman, Nived},
  booktitle={Proceedings of the Thirty-Eighth Conference on Learning Theory},
  series={Proceedings of Machine Learning Research},
  volume={291},
  pages={169--199},
  year={2025}
}

@article{adsul2021fast,
  title={Fast Algorithms for Rank-1 Bimatrix Games},
  author={Adsul, Bharat and Garg, Jugal and Mehta, Ruta and Sohoni, Milind and von Stengel, Bernhard},
  journal={Operations Research},
  volume={69},
  number={2},
  pages={613--631},
  year={2021}
}

@inproceedings{patris2024learning,
  title={Learning {Nash} Equilibria in Rank-1 Games},
  author={Patris, Nikolas and Panageas, Ioannis},
  booktitle={International Conference on Learning Representations},
  year={2024}
}

@article{kannan2010games,
  title={Games of Fixed Rank: A Hierarchy of Bimatrix Games},
  author={Kannan, Ravi and Theobald, Thorsten},
  journal={Economic Theory},
  volume={42},
  number={1},
  pages={157--173},
  year={2010}
}

@inproceedings{deng2019strategizing,
  title={Strategizing against No-Regret Learners},
  author={Deng, Yuan and Schneider, Jon and Sivan, Balasubramanian},
  booktitle={Advances in Neural Information Processing Systems},
  volume={32},
  year={2019}
}

@article{motzkin1965maxima,
  title={Maxima for Graphs and a New Proof of a Theorem of {Tur{\'a}n}},
  author={Motzkin, Theodore S. and Straus, Ernst G.},
  journal={Canadian Journal of Mathematics},
  volume={17},
  pages={533--540},
  year={1965}
}

@incollection{karp1972reducibility,
  title={Reducibility among Combinatorial Problems},
  author={Karp, Richard M.},
  booktitle={Complexity of Computer Computations},
  editor={Miller, Raymond E. and Thatcher, James W.},
  pages={85--103},
  publisher={Plenum Press},
  address={New York},
  year={1972}
}

@article{hubinger2019risks,
  title={Risks from Learned Optimization in Advanced Machine Learning Systems},
  author={Hubinger, Evan and van Merwijk, Chris and Mikulik, Vladimir and Skalse, Joar and Garrabrant, Scott},
  journal={arXiv preprint arXiv:1906.01820},
  year={2019}
}

@article{everitt2021reward,
  title={Reward Tampering Problems and Solutions in Reinforcement Learning: A Causal Influence Diagram Perspective},
  author={Everitt, Tom and Hutter, Marcus and Kumar, Ramana and Krakovna, Victoria},
  journal={Synthese},
  volume={198},
  pages={6435--6467},
  year={2021},
  doi={10.1007/s11229-021-03141-4}
}

@article{hadfieldmenell2016offswitch,
  title={The Off-Switch Game},
  author={Hadfield-Menell, Dylan and Dragan, Anca and Abbeel, Pieter and Russell, Stuart J.},
  journal={arXiv preprint arXiv:1611.08219},
  year={2016}
}

@inproceedings{hadfieldmenell2016cirl,
  title={Cooperative Inverse Reinforcement Learning},
  author={Hadfield-Menell, Dylan and Russell, Stuart J. and Abbeel, Pieter and Dragan, Anca D.},
  booktitle={Advances in Neural Information Processing Systems},
  volume={29},
  pages={3909--3917},
  year={2016}
}

@inproceedings{greenblatt2024control,
  title={{AI} Control: Improving Safety Despite Intentional Subversion},
  author={Greenblatt, Ryan and Shlegeris, Buck and Sachan, Kshitij and Roger, Fabien},
  booktitle={Proceedings of the 41st International Conference on Machine Learning},
  series={Proceedings of Machine Learning Research},
  volume={235},
  pages={16295--16336},
  year={2024}
}

@article{overman2025oversight,
  title={The Oversight Game: Learning to Cooperatively Balance an {AI} Agent's Safety and Autonomy},
  author={Overman, William and Bayati, Mohsen},
  journal={arXiv preprint arXiv:2510.26752},
  year={2025}
}

@inproceedings{braverman2018selling,
  title={Selling to a No-Regret Buyer},
  author={Braverman, Mark and Mao, Jieming and Schneider, Jon and Weinberg, S. Matthew},
  booktitle={Proceedings of the 19th ACM Conference on Economics and Computation},
  pages={523--538},
  publisher={ACM},
  year={2018},
  doi={10.1145/3219166.3219233}
}

@inproceedings{mansour2022bayesian,
  title={Strategizing against Learners in Bayesian Games},
  author={Mansour, Yishay and Mohri, Mehryar and Schneider, Jon and Sivan, Balasubramanian},
  booktitle={Proceedings of the 35th Conference on Learning Theory},
  series={Proceedings of Machine Learning Research},
  volume={178},
  pages={5221--5252},
  year={2022}
}

@inproceedings{cai2023multiple,
  title={Selling to Multiple No-Regret Buyers},
  author={Cai, Linda and Weinberg, S. Matthew and Wildenhain, Evan and Zhang, Shirley},
  booktitle={Proceedings of the 19th International Conference on Web and Internet Economics},
  pages={113--129},
  publisher={Springer},
  year={2023}
}

@article{guruganesh2024contracting,
  title={Contracting with a Learning Agent},
  author={Guruganesh, Guru and Kolumbus, Yoav and Schneider, Jon and Talgam-Cohen, Inbal and Vlatakis-Gkaragkounis, Emmanouil-Vasileios and Wang, Joshua R. and Weinberg, S. Matthew},
  journal={arXiv preprint arXiv:2401.16198},
  year={2024}
}

@inproceedings{brown2023learning,
  title={Is Learning in Games Good for the Learners?},
  author={Brown, William and Schneider, Jon and Vodrahalli, Kiran},
  booktitle={Advances in Neural Information Processing Systems},
  volume={36},
  pages={54228--54249},
  year={2023}
}

@article{guo2023hedge,
  title={The Optimal Strategy against Hedge Algorithm in Repeated Games},
  author={Guo, Xinxiang and Mu, Yifen},
  journal={arXiv preprint arXiv:2312.09472},
  year={2023}
}

@inproceedings{lipton2003playing,
  title={Playing Large Games Using Simple Strategies},
  author={Lipton, Richard J. and Markakis, Evangelos and Mehta, Aranyak},
  booktitle={Proceedings of the 4th ACM Conference on Electronic Commerce},
  pages={36--41},
  publisher={ACM},
  year={2003},
  doi={10.1145/779928.779933}
}

@incollection{theobald2009enumerating,
  title={Enumerating the Nash Equilibria of Rank 1-Games},
  author={Theobald, Thorsten},
  booktitle={Polyhedral Computation},
  editor={Avis, David and Bremner, David and Deza, Antoine},
  series={CRM Proceedings and Lecture Notes},
  volume={48},
  pages={115--130},
  publisher={American Mathematical Society},
  year={2009}
}

@inproceedings{adsul2011homeomorphism,
  title={Rank-1 Bimatrix Games: A Homeomorphism and a Polynomial Time Algorithm},
  author={Adsul, Bharat and Garg, Jugal and Mehta, Ruta and Sohoni, Milind A.},
  booktitle={Proceedings of the 43rd Annual ACM Symposium on Theory of Computing},
  pages={195--204},
  publisher={ACM},
  year={2011},
  doi={10.1145/1993636.1993664}
}

@article{vonstengel2012exponential,
  title={Rank-1 Games with Exponentially Many Nash Equilibria},
  author={von Stengel, Bernhard},
  journal={arXiv preprint arXiv:1211.2405},
  year={2012}
}

@inproceedings{garg2011bilinear,
  title={Bilinear Games: Polynomial Time Algorithms for Rank Based Subclasses},
  author={Garg, Jugal and Jiang, Albert Xin and Mehta, Ruta},
  booktitle={Internet and Network Economics},
  series={Lecture Notes in Computer Science},
  volume={7090},
  pages={399--407},
  publisher={Springer},
  year={2011},
  doi={10.1007/978-3-642-25510-6_35}
}

@article{mehta2018constant,
  title={Constant Rank Two-Player Games Are {PPAD}-Hard},
  author={Mehta, Ruta},
  journal={SIAM Journal on Computing},
  volume={47},
  number={5},
  pages={1858--1887},
  year={2018},
  doi={10.1137/15M1032338}
}

@inproceedings{anagnostides2022lastiterate,
  title={On Last-Iterate Convergence Beyond Zero-Sum Games},
  author={Anagnostides, Ioannis and Panageas, Ioannis and Farina, Gabriele and Sandholm, Tuomas},
  booktitle={Proceedings of the 39th International Conference on Machine Learning},
  series={Proceedings of Machine Learning Research},
  volume={162},
  pages={536--581},
  year={2022}
}

@inproceedings{anagnostides2022omd,
  title={Optimistic Mirror Descent Either Converges to Nash or to Strong Coarse Correlated Equilibria in Bimatrix Games},
  author={Anagnostides, Ioannis and Farina, Gabriele and Panageas, Ioannis and Sandholm, Tuomas},
  booktitle={Advances in Neural Information Processing Systems},
  volume={35},
  pages={16439--16454},
  year={2022}
}

\clearpage
\appendix
\section{Appendix}

\subsection{Strategic curvature and adjacent swaps}\label{app:curvature}

The exterior derivative of the one-form in \cref{eq:oneform} is
\begin{equation}
    \dd\omega
    =\sum_{1\le i<j\le n}\Omega_{ij}(z)\,\dd z_i\wedge\dd z_j,
    \qquad
    \Omega_{ij}(z)=u_j\partial_i g(z)-u_i\partial_j g(z).
    \label{eq:omegaij}
\end{equation}
We call $\Omega$ the \emph{strategic curvature}. In matrix form,
\begin{equation}
    \Omega(z)=\nabla g(z)u^\top-u\nabla g(z)^\top.
    \label{eq:curvature-matrix}
\end{equation}

\begin{proposition}[Strategic curvature]
\label{prop:curvature}
For $y=y(z)$,
\begin{equation}
    \nabla g(z)
    =\eta B\bigl(\diag(y)-yy^\top\bigr)v.
    \label{eq:gradg}
\end{equation}
Consequently, $\rank(\Omega(z))\le2$ for every $z$.
\end{proposition}

\begin{proof}
The Jacobian of the softmax map at $h$ is
\[
    \eta\bigl(\diag(y)-yy^\top\bigr).
\]
Applying the chain rule to $g(z)=v^\top y(z)$ gives \cref{eq:gradg}. Equation \eqref{eq:curvature-matrix} is the difference of two rank-one matrices, so its rank is at most two.
\end{proof}

The curvature has a direct operational meaning. Fix a cumulative state $z\in\mathcal D_T$ and two optimizer actions $i,j$. For durations $a,b\ge0$, compare the schedules that play $i$ for time $a$ then $j$ for time $b$, and the reverse order. Their terminal cumulative play is identical, so their endpoint terms cancel.

\begin{proposition}[Adjacent-swap identity]
\label{prop:swap}
Let $\Delta_{ij}^{a,b}(z)$ denote the total reward from playing $i$ for duration $a$ and then $j$ for duration $b$, minus the reward from the reverse order, starting from cumulative state $z$ and followed, if desired, by the same continuation. Whenever the corresponding rectangle lies in $\mathcal D_T$,
\begin{equation}
\boxed{
    \Delta_{ij}^{a,b}(z)
    =\int_0^a\int_0^b
      \Omega_{ij}(z+s e_i+t e_j)\dd t\dd s.}
\label{eq:swap}
\end{equation}
\end{proposition}

\begin{proof}
Both orders reach the same terminal cumulative state, so the $\Phi_\eta$ endpoint terms in \cref{eq:path-form} cancel. The remaining difference is the circulation of the one-form \eqref{eq:oneform} around the oriented rectangle with sides $ae_i$ and $be_j$. Green's theorem yields \cref{eq:swap}.
\end{proof}

Thus $\Omega_{ij}>0$ means that, locally and holding occupation times fixed, action $i$ should precede action $j$.

\subsection{Path-independent rank-one games}\label{app:pathind}

\begin{theorem}[Path-independent rank-one games]
\label{thm:pathind}
Assume $u\ne0$. On the convex domain $\mathcal D_T$, the following are equivalent:
\begin{enumerate}[label=(\roman*)]
    \item the line integral $\int_\gamma g(z)u^\top\dd z$ depends only on the endpoints of $\gamma$;
    \item $\Omega(z)=0$ for all $z$ in the relative interior of $\mathcal D_T$;
    \item $\nabla g(z)\in\operatorname{span}(u)$ for all such $z$;
    \item there is a scalar function $\bar g$ on $u^\top\mathcal D_T$ such that $g(z)=\bar g(u^\top z)$.
\end{enumerate}
If these conditions hold and $G'=\bar g$, then
\begin{equation}
\boxed{
    \Rstar(h^0,T,A,B)
    =\max_{q\in T\Delta_n}
      \left\{
      \Phi_\eta(h^0)-\Phi_\eta(h^0+B^\top q)
      +G(u^\top q)-G(0)
      \right\}.}
    \label{eq:pathind-value}
\end{equation}
In particular, an optimal strategy can be chosen constant, $x(t)=q/T$.
\end{theorem}

\begin{proof}
On a convex domain, a continuously differentiable one-form is path-independent if and only if it is closed. The closedness condition is precisely $\Omega_{ij}=0$ for all $i,j$. If $u\ne0$, this is equivalent to $\nabla g$ being parallel to $u$. Hence $g$ is constant along every connected slice orthogonal to $u$, giving $g(z)=\bar g(u^\top z)$. The converse is immediate. The line integral then equals $G(u^\top q)-G(0)$, and every $q\in T\Delta_n$ is reached by the constant control $q/T$.
\end{proof}

Two useful algebraic subclasses follow.

\begin{corollary}[Column-independent rank-one component]
\label{cor:vconstant}
If $v=c\one$ for some $c\in\R$, then $g(z)=c$, and
\begin{equation}
    \Rstar(h^0,T,A,B)
    =\max_{q\in T\Delta_n}
      \left\{
      \Phi_\eta(h^0)-\Phi_\eta(h^0+B^\top q)+c\,u^\top q
      \right\}.
    \label{eq:vconstant}
\end{equation}
The objective is concave in $q$, so an additive approximation can be computed by standard convex optimization.
\end{corollary}

\begin{corollary}[One-dimensional learner state]
\label{cor:onedim}
Let $P=I_m-\frac1m\one\one^\top$ be the centering projection. If
\[
    BP=uw^\top
\]
for some $w\in\R^m$, then $g(z)$ depends only on $u^\top z$, and \cref{thm:pathind} applies.
\end{corollary}

\begin{proof}
Softmax is invariant to adding a multiple of $\one$. The centered learner state is
\[
    P(h^0+B^\top z)=Ph^0+(BP)^\top z=Ph^0+w(u^\top z),
\]
which depends on $z$ only through $u^\top z$.
\end{proof}

\subsection{Two-action learners and block approximation}
\label{app:two-actions}

Suppose the learner has $m=2$ actions. Define
\begin{equation}
    d_i=B_{i1}-B_{i2},
    \qquad
    \delta=v_1-v_2,
    \qquad
    r(z)=h^0_1-h^0_2+d^\top z.
    \label{eq:two-defs}
\end{equation}
Writing $\sigma(s)=(1+\e^{-s})^{-1}$, we have
\[
    y_1(z)=\sigma(\eta r(z)),
    \qquad
    y_2(z)=1-y_1(z),
\]
and
\begin{equation}
    g(z)=v_2+\delta\sigma(\eta r(z)).
\end{equation}
Therefore
\begin{equation}
\boxed{
    \Omega_{ij}(z)
    =\eta\delta\,\sigma(\eta r(z))\bigl(1-\sigma(\eta r(z))\bigr)
      (u_jd_i-u_id_j).}
    \label{eq:two-curvature}
\end{equation}
Its sign is independent of the learner state.

\begin{lemma}[Pure block approximation and attainment]
\label{lem:block-approx}
Let $x:[0,T]\to\Delta_n$ be an admissible control and let
\[
    \tau_i=\int_0^T x_i(t)\dd t.
\]
There is a sequence of pure, piecewise-constant controls $x^{(k)}$ with the same occupation times $(\tau_i)_{i=1}^n$ such that
\[
    \sup_{t\in[0,T]}\|z_{x^{(k)}}(t)-z_x(t)\|_1\longrightarrow0
    \quad\text{and}\quad
    \Rcont(x^{(k)},h^0,T,A,B)\longrightarrow\Rcont(x,h^0,T,A,B).
\]
Moreover, for every fixed ordering $\pi$ of the optimizer actions, the reward $J_\pi(\tau)$ of the contiguous block schedule with block lengths $\tau\in T\Delta_n$ is continuous in $\tau$. Hence $J_\pi$ attains a maximum on $T\Delta_n$.
\end{lemma}

\begin{proof}
Partition $[0,T]$ into intervals of mesh at most $1/k$. On each interval $I$, replace $x$ by a pure schedule that plays action $i$ for exactly $\int_I x_i(t)\dd t$ units of time, in any fixed local order. The original and purified cumulative paths agree at every partition endpoint, and inside an interval their $\ell_1$ distance is at most twice the interval length. This gives uniform convergence of the cumulative paths and preserves every total occupation time exactly.

The map $z\mapsto y(z)=\smx{h^0+B^\top z}$ is Lipschitz on the compact domain $\mathcal D_T$. On each partition interval, freeze the learner strategy at its value at the left endpoint. The original and purified controls have the same integrated action vector on that interval, so their rewards against the frozen learner strategy are equal. The error from unfreezing is $O(|I|^2)$ for each control, uniformly over intervals. Summing gives an $O(1/k)$ reward difference and proves convergence.

For a contiguous block schedule in a fixed order, its cumulative path varies uniformly with the block-length vector, apart from shifts of finitely many breakpoints; the integrand is bounded and continuous in the cumulative state. Dominated convergence therefore gives continuity of $J_\pi$. The simplex $T\Delta_n$ is compact, so the maximum is attained.
\end{proof}

\begin{theorem}[Global ordering for a two-action learner]
\label{thm:sorting}
Assume $m=2$ and that the coordinates of $u$ have a common nonzero sign. Normalize the factorization $A+B=uv^\top$ so that $u_i>0$ for every optimizer action. If $\delta>0$, then for every fixed vector of occupation times $(\tau_i)_{i=1}^n$, the contiguous block schedule that plays the used actions in nonincreasing order of
\begin{equation}
    \rho_i=\frac{d_i}{u_i}.
    \label{eq:rho}
\end{equation}
is optimal among all admissible controls with those occupation times. If $\delta<0$, the order is reversed; if $\delta=0$, the ordering is immaterial. Consequently, there is a globally optimal strategy in which every used optimizer action appears in one contiguous block and the total number of switches is at most $n-1$.
\end{theorem}

\begin{proof}
Suppose $\delta>0$ and fix occupation times $\tau$. By \cref{lem:block-approx}, any admissible control with occupation vector $\tau$ is the reward limit of finite pure schedules with exactly the same occupation times. If $\rho_i\ge\rho_j$, then $u_jd_i-u_id_j\ge0$, so \cref{eq:two-curvature,prop:swap} show that swapping any adjacent out-of-order occurrence of $j$ before $i$ cannot decrease reward. Repeated adjacent swaps transform each approximating schedule into the same canonical contiguous block schedule sorted by nonincreasing $\rho_i$. Passing to the limit proves that this block schedule is optimal for the fixed occupation vector.

Its reward is continuous in $\tau$ by \cref{lem:block-approx}; therefore it attains a maximum over the compact simplex $T\Delta_n$. This maximizer is globally optimal. The case $\delta<0$ is symmetric, and for $\delta=0$ all swap differences vanish.
\end{proof}

For two optimizer actions, no common-sign assumption is needed.

\begin{corollary}[One-switch structure in $2\times2$ games]
\label{cor:one-switch}
Suppose $n=m=2$, and define
\[
    \chi=\delta(u_2d_1-u_1d_2).
\]
There is an optimal strategy of the form
\[
    1\to2 \quad\text{if }\chi>0,
    \qquad
    2\to1 \quad\text{if }\chi<0.
\]
If $\chi=0$, the ordering is immaterial. Thus an optimal strategy switches at most once.
\end{corollary}

\begin{proof}
For fixed occupation times $(\tau,T-\tau)$, \cref{eq:two-curvature} has the constant sign of $\chi$. The adjacent-swap identity and \cref{lem:block-approx} therefore show that the corresponding preferred block order is optimal among all controls with those occupation times. Its reward is continuous in $\tau\in[0,T]$, so a maximizing switch time exists. If $\chi=0$, every adjacent swap has zero value and either block order is optimal.
\end{proof}

For completeness, the remaining one-dimensional optimization can be written explicitly. Let
\[
    L_\eta(r)=\frac1\eta\log(1+\e^{\eta r})
\]
and define
\begin{equation}
    \Psi_\eta(r,d,s)
    =\begin{cases}
      \dfrac{L_\eta(r+ds)-L_\eta(r)}{d},&d\ne0,\\[2mm]
      s\,\sigma(\eta r),&d=0.
    \end{cases}
    \label{eq:Psi}
\end{equation}
If pure row $i$ is played for duration $s$ starting from log-odds state $r$, the accumulated reward is
\begin{equation}
    \mathcal J_i(r,s)
    =A_{i2}s+(A_{i1}-A_{i2})\Psi_\eta(r,d_i,s).
    \label{eq:block-reward}
\end{equation}
Hence, if the preferred ordering is $i\to j$,
\begin{equation}
    \Rstar(h^0,T,A,B)
    =\max_{\tau\in[0,T]}
      \left\{
        \mathcal J_i(r_0,\tau)
        +\mathcal J_j(r_0+d_i\tau,T-\tau)
      \right\},
    \label{eq:2x2-value}
\end{equation}
where $r_0=h^0_1-h^0_2$.

\subsection{Train-then-harvest calculations}
\label{app:train-harvest}

We first justify the reduction of \cref{ex:train-harvest} to the scalar objective in \cref{eq:train-reward}. Let $a(t)=x_1(t)$ and
\[
    q(t)=\int_0^t a(s)\dd s.
\]
The learner's probability of column two is $\sigma(\eta q(t))$, and the optimizer's reward is
\begin{equation}
    R[a]=\int_0^T(1-a(t))\sigma(\eta q(t))\dd t.
    \label{eq:train-functional}
\end{equation}
For fixed $q(T)=\tau$,
\[
    R[a]
    =\int_0^T\sigma(\eta q(t))\dd t
     -\int_0^\tau\sigma(\eta s)\dd s.
\]
The second term is fixed. Since $0\le\dot q\le1$, $q$ is nondecreasing, and $q(T)=\tau$, we have $q(t)\le\min\{t,\tau\}$ for every $t$. The first term is therefore maximized by $q(t)=\min\{t,\tau\}$, which corresponds to training on row one for time $\tau$ and then harvesting on row two. Substitution gives \cref{eq:train-reward}.

\begin{proof}[Proof of \cref{prop:lambert}]
Differentiating \cref{eq:train-reward},
\[
    R'(\tau)
    =\sigma(\eta\tau)
      \left[-1+\eta(T-\tau)(1-\sigma(\eta\tau))\right].
\]
The bracketed term is strictly decreasing. At zero it equals $-1+\eta T/2$, giving the boundary case. In the interior, the first-order condition is
\[
    \eta(T-\tau)=1+\e^{\eta\tau}.
\]
Letting $w=\e^{\eta\tau}$ gives $w+\log w=\eta T-1$, hence $w=W(\e^{\eta T-1})$. At the optimum,
\[
    T-\tau^*=\frac{1+w}{\eta},
    \qquad
    \sigma(\eta\tau^*)=\frac{w}{1+w},
\]
so the value collapses to
\[
    R(\tau^*)
    =\frac{1+w}{\eta}\cdot\frac{w}{1+w}
    =\frac{w}{\eta}.
\]
This gives \cref{eq:lambert}.
\end{proof}

\begin{proof}[Proof of \cref{prop:linear-gap}]
For a constant strategy with $x_1(t)=p>0$,
\[
    R_{\mathrm{const}}(p)
    =\frac{1-p}{\eta p}
      \log\left(\frac{1+\e^{\eta pT}}2\right).
\]
Let $s=Lp$. Since
\[
    \log\left(\frac{1+\e^s}{2}\right)
    =\frac s2+\log\cosh\frac s2
    \le\frac s2+\frac{s^2}{8},
\]
we obtain, for $L\le4$,
\[
    \frac{R_{\mathrm{const}}(p)}T
    \le\left(1-\frac sL\right)\left(\frac12+\frac s8\right)
    \le\frac12.
\]
Equality is attained at $p=0$. The dynamic value follows from \cref{prop:lambert}, and strict positivity follows because $L>2$ gives an interior optimum.
\end{proof}

\subsection{Hamilton--Jacobi and Pontryagin formulations}
\label{app:control}

Let $V(\tau,h)$ be the optimal reward over a remaining horizon $\tau$ from learner state $h$. The Hamilton--Jacobi--Bellman equation is
\begin{equation}
    \partial_\tau V(\tau,h)
    =\max_{i\in[n]}
      \left\{(Ay(h))_i+(B\nabla_hV(\tau,h))_i\right\},
    \qquad
    V(0,h)=0.
    \label{eq:hjb}
\end{equation}
For a rank-one game, substitute $V(\tau,h)=\Phi_\eta(h)+U(\tau,h)$. Since $Ay=-By+u(v^\top y)$ and $\nabla\Phi_\eta=y$, the zero-sum terms cancel:
\begin{equation}
    \partial_\tau U(\tau,h)
    =\max_{i\in[n]}
      \left\{u_i v^\top y(h)+(B\nabla_hU(\tau,h))_i\right\},
    \qquad
    U(0,h)=-\Phi_\eta(h).
    \label{eq:rankone-hjb}
\end{equation}

In cumulative-play coordinates, omitting the constant $\Phi_\eta(h^0)$, the objective is
\[
    -\Phi_\eta(h^0+B^\top z(T))
    +\int_0^T(u^\top x(t))g(z(t))\dd t.
\]
The Pontryagin Hamiltonian is
\begin{equation}
    \mathcal H(z,p,x)=x^\top(p+ug(z)).
\end{equation}
The costate and terminal conditions are
\begin{equation}
    \dot p(t)=-(u^\top x(t))\nabla g(z(t)),
    \qquad
    p(T)=-By(T).
    \label{eq:costate}
\end{equation}
An optimal control is supported on
\begin{equation}
    \arg\max_i\{p_i(t)+u_i g(z(t))\}.
\end{equation}
The switching scores $s_i=p_i+u_ig$ satisfy
\begin{equation}
    \dot s_i(t)
    =\sum_jx_j(t)\left(u_i\partial_jg-u_j\partial_ig\right),
\end{equation}
showing again that switching is driven exactly by the curvature coefficients.

\subsection{A concave subclass for two-action learners}
\label{app:concave}

Assume $m=2$, $\delta=v_1-v_2>0$, and $u_i,d_i>0$. Relabel optimizer actions so that
\[
    \rho_1=\frac{d_1}{u_1}\ge\rho_2\ge\cdots\ge\rho_n.
\]
By \cref{thm:sorting}, an optimal strategy uses this block order. Let $\tau_i\ge0$ be block lengths summing to $T$, and define
\[
    r_k=r_0+\sum_{\ell=1}^k d_\ell\tau_\ell.
\]
Using $A_{i1}-A_{i2}=u_i\delta-d_i$, the objective is
\begin{align}
    J(\tau)
    ={}&L_\eta(r_0)-L_\eta(r_n)
       +\sum_{i=1}^nA_{i2}\tau_i\notag\\
      &+\delta\sum_{i=1}^n\frac{u_i}{d_i}
        \left[L_\eta(r_i)-L_\eta(r_{i-1})\right].
    \label{eq:concave-objective}
\end{align}
The coefficient of $L_\eta(r_i)$ for $i<n$ is
\[
    \delta\left(\frac1{\rho_i}-\frac1{\rho_{i+1}}\right)\le0,
\]
and the terminal coefficient is $\delta/\rho_n-1$. Therefore:

\begin{proposition}
If, in addition, $\rho_n\ge\delta$, then $J$ is concave on the simplex $\{\tau\ge0:\sum_i\tau_i=T\}$. Hence the optimal block lengths can be approximated in polynomial time by convex optimization.
\end{proposition}

\subsection{A weak-learning ordering expansion}
\label{app:small-eta}

Assume $h^0=0$ and $\eta\|B\|T$ is small. Let
\[
    \bar y=\frac1m\one,
    \qquad
    C_0=\diag(\bar y)-\bar y\bar y^\top,
    \qquad
    w=BC_0v.
\]
A Taylor expansion of softmax gives, uniformly on $\mathcal D_T$,
\begin{equation}
    y(z)=\bar y+\eta C_0B^\top z+O(\eta^2\|B\|^2T^2),
\end{equation}
and therefore
\begin{equation}
    g(z)=v^\top\bar y+\eta w^\top z+O(\eta^2\|B\|^2\|v\|T^2).
\end{equation}
For a block schedule with fixed occupation times $q_i$ and order $\pi$, the first-order path-dependent term is
\begin{align}
    \int_\gamma g(z)u^\top\dd z
    ={}& (v^\top\bar y)u^\top q
       +\frac\eta2\sum_i u_iw_iq_i^2\notag\\
      &+\eta\sum_{r<s}w_{\pi_r}u_{\pi_s}q_{\pi_r}q_{\pi_s}
       +O\!\left(\eta^2\|B\|^2\|u\|\|v\|T^3\right).
    \label{eq:small-eta-order}
\end{align}
Thus the first-order benefit of placing $i$ before $j$ is
\[
    w_i u_j-w_j u_i.
\]
When all $u_i>0$, the first-order optimal ordering is nonincreasing in $w_i/u_i$. This generalizes the exact $m=2$ ordering rule at weak learning rates.

\subsection{Discrete-time MWU identity}
\label{app:discrete}

Consider the discrete update
\[
    h_{t+1}=h_t+B^\top x_t,
    \qquad
    y_t=\nabla\Phi_\eta(h_t).
\]
The Bregman identity gives
\[
    x_t^\top By_t
    =\Phi_\eta(h_{t+1})-\Phi_\eta(h_t)
     -D_{\Phi_\eta}(h_{t+1},h_t).
\]
For log-sum-exp,
\[
    D_{\Phi_\eta}(h_{t+1},h_t)
    =\frac1\eta\KL(y_t\|y_{t+1}).
\]
Hence, whenever $A+B=uv^\top$,
\begin{proposition}[Discrete rank-one decomposition]
\begin{equation}
\boxed{
\begin{aligned}
    \sum_{t=0}^{T-1}x_t^\top Ay_t
    ={}&\Phi_\eta(h_0)-\Phi_\eta(h_T)
       +\sum_{t=0}^{T-1}(u^\top x_t)(v^\top y_t)\\
      &+\frac1\eta\sum_{t=0}^{T-1}\KL(y_t\|y_{t+1}).
\end{aligned}}
\end{equation}
\end{proposition}
The final nonnegative term has no continuous-time counterpart. Even in zero-sum games, it rewards movement of the learner's strategy and is the source of additional discrete-time path dependence.

\end{document}